\documentclass{llncs}

\usepackage{amsmath}
\usepackage{cite}
\usepackage{amssymb}
\usepackage{xspace}
\usepackage{pst-all,xcolor}

\newcommand{\NN}{\ensuremath{\mathbb{N}}}
\newcommand{\RR}{\ensuremath{\mathbb{R}}}

\newcommand{\junk}[1]{}

\title{Considerate Equilibrium\thanks{This work was supported by the German Israeli Foundation (GIF) under contract 877/05 and by DFG grant Ho 3831/3-1.}}
\titlerunning{Considerate Equilibrium}

\author{Martin Hoefer\inst{1} \and Michal Penn\inst{2} \and Maria
  Polukarov\inst{3} \and\\ Alexander Skopalik\inst{4} \and Berthold
  V\"ocking\inst{1}} 
\institute{Dept.\ of Computer Science, RWTH Aachen University, Germany\\
  \email{\{mhoefer,voecking\}@cs.rwth-aachen.de} 
  \and Faculty of Industrial Engineering and Management, Technion, Israel\\
  \email{mpenn@ie.technion.ac.il}
  \and School of Electronics and Computer Science, University of Southampton, UK\\
  \email{mp3@ecs.soton.ac.uk} 
  \and School of Physical and Mathematical Sciences, Nanyang Technological University Singapore\\
  \email{skopalik@cs.rwth-aachen.de}
}

\begin{document}
\maketitle
 
\begin{abstract}
  We consider the existence and computational complexity of
  coalitional stability concepts based on social networks. Our
  concepts represent a natural and rich combinatorial generalization
  of a recent approach termed \emph{partition
    equilibrium}~\cite{Feldman09}. We assume that players in a
  strategic game are embedded in a social network, and there are
  coordination constraints that restrict the potential coalitions that
  can jointly deviate in the game to the set of cliques in the social
  network. In addition, players act in a ``considerate'' fashion to
  ignore potentially profitable (group) deviations if the change in
  their strategy may cause a decrease of utility to their neighbors.

  We study the properties of such \emph{considerate equilibria} in
  application to the class of \emph{resource selection games (RSG)}.
  Our main result proves existence of a \emph{considerate equilibrium}
  in all symmetric RSG with strictly increasing delays, for \emph{any}
  social network among the players. The existence proof is constructive
  and yields an efficient algorithm. In fact, the computed considerate
  equilibrium is a Nash equilibrium for the standard RSG
  showing that there exists a state that is stable against selfish and
  considerate behavior simultaneously. In addition, we show results
  on convergence of considerate dynamics.


\end{abstract}

\section{Introduction}
Game theory provides tools for the analysis of the outcome of social
interaction of self-motivated, rational agents. Rationality is usually
captured in a way that agents are acting autonomously in order to
maximize a utility function. This leads to much interest in the study
of stable outcomes in games, making it the central topic in game
theory. In strategic games the standard concept of stability is the
{\it Nash equilibrium (NE)} -- a state resilient to unilateral
strategy changes of players. While a mixed Nash equilibrium is
guaranteed to exist, a pure Nash equilibrium might not exist in
general, though has been proven to exist in several interesting
classes such as {\it congestion
  games}~\cite{Monderer96,Rosenthal73}. A drawback of Nash equilibrium
is that it neglects coalitional deviations by groups of players; these
are captured most prominently by the notion of {\it strong equilibrium
  (SE)}~\cite{Aumann59}, in which no coalition can strictly improve
the utility of \emph{all} participants. A slightly stronger variant
termed {\it super-strong equilibrium
  (SSE)}~\cite{Rozenfeld07,Feldman09} guarantees that no coalition can
strictly improve \emph{any} participant without strictly deteriorating
at least one other participant. SSE postulates the natural and widely
considered condition of (strong) Pareto efficiency~\cite{Myerson04}
for every coalition. However, while stability against deviations by
coalitions of players is a most natural desideratum, it is well-known
that there are only very few strategic games with SE, and SSE are even
harder to guarantee.

In contrast to the assumptions underlying SE and SSE, many real-life
scenarios allow only certain subsets of players to cooperate because a
group of players has to find a deviation, agree on it, and coordinate
individual actions. This is impossible for a subset of players that
are completely unrelated to each other.
%
%
A promising recent approach for limited coalitional deviations was
studied prominently in resource selection games~\cite{Feldman09}. In
this case, there is a given partition of the set of players such that
only sets of the partition can implement coalitional deviations. The
power of this restriction was demonstrated on the concept of SSE - a
\emph{partition equilibrium} is a SSE subject to coalitional
deviations by player sets in the partition only. In contrast to SSE,
it was shown that partition equilibrium always exists in resource
selection games~\cite{AnshelevichSAGT10}, and that the profiles are
also NE - that is, coalitional and unilateral stability are obtained
simultaneously. The restriction of coalitional deviations in partition
equilibrium essentially postulates two structural properties: (1)
coalitions of players that execute a strategy change have to be close
to each other, and (2) their decision must strictly benefit at least
one of them but not strictly deteriorate any other player close to
them. The notion of closeness is defined in both cases simply as being
in the same partition.

In this paper, we significantly strengthen the partition equilibrium
concept by considering coalitional deviations and equilibria based a
rich combinatorial structure derived from a social network among the
players rather than just partitions. In our case, (1) coalitions of
players that execute a strategy change must be cliques in the graph,
and (2) their decision must not strictly deteriorate any neighboring
players. The solution concept naturally corresponding to considerate
behavior is the {\em considerate equilibrium}, i.e., a state in which
(1) no coaliton formed by a clique in the social network can deviate
so that the utility of at least one member of the coalition strictly
improves and (2) none of the players neighboring the clique gets
worse.
Observe that partition equilibrium evolves as a special case of
considerate equilibrium when the social network is composed of a set
of disjoint cliques. To the best of our knowledge, our approach has
not been considered before.



We study considerate behavior in the prominent class of {\it resource
  selection games (RSG)}. In an RSG, each player chooses one of a
finite set of resources, and its cost is given by a delay function
depending on the number of players choosing the resource. RSGs are a
fundamental setting in computer science, operations research and
economics, due to their practical applicability (e.g., in electronic
commerce and communication networks) and plausible analytical
properties. In particular, for strictly increasing delay functions, SE
always exist~\cite{Holzman97,Holzman03}, but SSE do not necessarily
exist~\cite{Feldman09}.
The latter fact is the motivation for studying the effects imposed by
natural restrictions to the coalitional structure on the existence of SSE
initiated by Feldman and Tennenholtz in~\cite{Feldman09}.

\subsection{Our Results}

We show that regardless of the social network, all RSGs with strictly
increasing delay functions possess a considerate equilibrium.
Our proof in Section~\ref{sec:exist} is constructive
and yields an efficient algorithm for computing such an
equilibrium. Indeed, the computed super-strong considerate equilibrium
is an NE for the standard RSG showing that there exists a state that is
stable against selfish and considerate behavior
simultaneously. Observe that the number of cliques might be
exponential in the number of players such that not even the
computation of a single improving move is non-trivial. We solve this
problem by showing that, in an NE, every profitable deviation of a
clique is witnessed by a move of a single player decreasing a suitably
defined potential function. In addition, our proof is fundamentally
different and significantly simpler than the proof for existence for
the special case of partition equilibrium in~\cite{AnshelevichSAGT10}.


In Section~\ref{sec:converge}, we consider convergence properties of
dynamics. Let us remark that the potential function approach from the
existence proof does not imply that the sequential dynamics defined by
deviations of cliques is acyclic, since the single player moves
considered in the existence proof do not necessarily correspond to
allowed improving moves. Indeed, we show that even for identical,
strictly increasing delays there are infinite sequences of improving
moves of cliques. This is in contrast to the dynamics corresponding to
partition equilibrium, for which we can show the finite improvement
property in this setting.

\subsection{Related work}

Using a social network approach to restrict coalitional deviations in
games, our work is related to an emerging area in social sciences,
game theory, and computer science. While the study of social
connections is central to social sciences, and the notion of stability
is central to game theory, a standard tool for analyzing the interplay
between social context and outcome of games has received attention
only recently. Perhaps most relevant in this spirit
are~\cite{Feldman09,AnshelevichSAGT10} on partition equilibrium
discussed above.

The notion of partition equilibrium is related to work on {\it social
  context games}~\cite{Ashlagi08}, where a player's utility can be
affected by the payoffs of other players. For example, a player may be
interested in ranking his payoff as high as possible comparing to the
others' payoffs~\cite{Brandt06}, or a player may care about the total
payoff of a subset of his ``friends'', as in {\it coalitional
  congestion games}~\cite{Hayrapetyan06,Kuniavsky07}. A social context
game is then defined by some underlying game, the social context given
by some topological or graph-theoretic structure of neighborhood, and
aggregation functions capturing the effects of utility changes in the
underlying game on player incentives. In~\cite{Ashlagi08}, RSG are
considered as the underlying games, and four natural social contexts
are studied. However, unlike for partition equilibrium, this work
deals only with unilateral deviations.

While~\cite{Ashlagi08,Feldman09} are initial steps in relating the
social structure to the outcome of a game, they are quite restrictive
in that only particular social contexts and fixed coalitional
structures (partitions) are considered. In addition, they ignore the
phenomenon of considerate behavior which is present in our
work. Similar arguments apply w.r.t.~\cite{Fotakis08}, where fixed
coalition structures in load balancing and congestion games are
studied. Here coalitions act as single ``splittable'' coalition
players that strive to minimize the makespan or the sum of costs of
the agents in the coalition.


\section{Preliminaries and Initial Results}

A \emph{strategic game} is a tuple $(N, (S_i)_{i \in N}, (u_i)_{i \in
  N})$, where $N$ is the set of $n$ \emph{players}, $S_i$ is a
\emph{strategy space} of player $i$. A \emph{state} $s$ of the game is
a vector of strategies $(s_1,\ldots,s_n)$, where $s_i \in S_i$. For
convenience, we use $s_{-i}$ to denote
$(s_1,\ldots,s_{i-1},s_{i+1},\ldots,s_n)$, i.e., $s$ reduced by the
single entry of player $i$. Similarly, for a state $s$ we use $s_C$ to
denote the strategy choices of a coalition $C \subseteq N$ and
$s_{-C}$ for the complement, and we write $s = (s_C,s_{-C})$. The
\emph{utility} of player $i$ in state $s$ is $u_i(s) \in \RR$. For a
state $s$ a coalition $C \subseteq N$ is said to have an
\emph{improving move} if there is $s'_C$ such that $u_i(s'_C,s_{-C}) >
u_i(s)$ for every player $i \in C$. In particular, the improving move
is \emph{unilateral} if $|C|=1$. A state has a \emph{weak improving
  move} if there is $C \subseteq N$ and $s'_{C}$ such that
$u_i(s'_C,s_{-C}) \ge u_i(s)$ for every $i \in C$ and
$u_i(s'_C,s_{-C}) > u_i(s)$ for \emph{at least one} $i \in C$. A
\emph{(pure) Nash equilibrium (NE)}~\cite{Nash51} is a state that has
no unilateral improving moves, a \emph{strong equilibrium
  (SE)}~\cite{Aumann59} a state that has no improving moves, and a
\emph{super-strong equilibrium (SSE)}~\cite{Feldman09} a state that
has no weak improving moves.

To model considerate behavior, we adjust the definition of improving
moves. In particular, there is an undirected, unweighted graph
$G=(N,E)$ over the set of players. For a subset $C \subseteq N$
consider the \emph{neighborhood} of $C$ as $\mathcal{N}(C) = \{j \in N
\mid \exists i \in C, \{i,j\} \in E\}$.
\begin{definition}[Considerate Improving Moves]
  A state $s$ has a \emph{considerate improving move} for a coalition
  $C$ if there is $s'_C$ such that $u_i(s'_C,s_{-C}) > u_i(s)$ for all
  $i \in C$ and $u_j(s'_C,s_{-C}) \ge u_j(s)$ for all $j \in
  \mathcal{N}(C)$. For a \emph{unilateral considerate improving move}
  we have $|C|=1$. A state $s$ has a \emph{weak considerate improving
    move} for a coalition $C$ if there is $s'_C$ such that
  $u_i(s'_C,s_{-C}) \ge u_i(s)$ for all $i \in C \cup \mathcal{N}(C)$
  and $u_i(s'_C,s_{-C}) > u_i(s)$ for at least one $i \in C$.
\end{definition}
Note that every (weak/unilateral) considerate improving move is also a
(weak/uni\-lateral) improving move but not vice versa. To define
coalitional equilibria, let us, for the time being, also assume that
there is a set system of \emph{feasible coalitions} $\mathcal{C}
\subseteq 2^N$.  A \emph{considerate Nash equilibrium (CNE)} is a
state $s$ that has no unilateral considerate improving moves. A
\emph{(super) strong considerate equilibrium ((S)SCE)} is a state $s$
that has no (weak) considerate improving move for a coalition $C \in
\mathcal{C}$. Note that for CNE we implicitly assume $\mathcal{C}$ is
the set of all singleton sets $\{i\}$ for all $i \in N$. Every NE is a
CNE, and every (S)SE is a (S)SCE. The converse only holds for CNE and
NE if $E = \emptyset$. In general SCE and SSCE are SE and SSE only if
$E = \emptyset$ and $\mathcal{C} = 2^N$, respectively. In this way,
existence of social ties and a non-trivial set of feasible coalitions
weaken the structural requirements for existence of equilibrium.

In the rest of the paper, we make the natural assumption that the set
of feasible coalitions corresponds to the set of cliques in $G$. In our
analysis, we focus on weak improving moves and study super strong
considerate equilibria as we believe that this solution concept is
most interesting not only from a technical point of view but also a
natural and convincing model for the interaction of coalitional
structures in the presence of a social network.

\begin{definition}[Considerate Equilibria]
  A \emph{considerate equilibrium (CE)} is a state $s$ that has no
  weak considerate improving move for a coalition corresponding to a
  clique in $G$.
\end{definition}

Note that CE nicely generalizes partition equilibrium. In particular,
a partition equilibrium is a CE if the social network $G$ is
partitioned into isolated cliques. Note that we do not explicitely
assume that the set of feasible coalitions is restricted to maximal
cliques. If the graph is partitioned into isolated cliques, however,
this rather technical assumption made in the definition of partition
equilibrium is a natural consequence of the assumption that the
coalitions behave considerately since one can assume w.l.o.g. that all
members of a partition participate in a coalition as weak improving
moves do not decrease the utility of neighboring players.  \bigskip

\emph{Resource selection games (RSG)}, sometimes referred to as
\emph{singleton congestion} or \emph{parallel link games}, are a basic
class of potential games. There is a set of resources $R$ and $S_i =
R$ for every player $i \in N$. For a state $s$ we denote by
$\ell_r(s)$ the number of players that pick $r \in R$ in $s$. For each
resource $r \in R$ there is a delay function $d_r(x) \in
\NN$. Throughout the paper we assume that all delay functions are
non-negative and strictly increasing. In a state with $s_i = r$,
player $i$ has \emph{cost} $c_i(s) = -u_i(s) = d_r(\ell_r(s))$.

In this paper, we consider RSGs with strictly increasing delays. In
this case, it is known that NE exist~\cite{Rosenthal73}, can be
computed in polynomial time~\cite{Fotakis09}, and are equivalent to
SE~\cite{Holzman97}. Moreover, the games possess a (strong) potential
function~\cite{Monderer96,Holzman97}, i.e., every sequence of
unilateral improving moves has finite length and ends in a
NE/SE. Trivially, by restriction of improving moves, the same holds
also for CNE and SCE. Interestingly, however, even in simplest games
SSE are not guaranteed to exist\footnote{Consider a game with
  $N=\{1,2,3\}$, $R=\{r_1,r_2\}$, and $d_{r_1}(x) = d_{r_2}(x) =
  x$.}. In contrast, we show below that CE always exist. However, even
for identical resources we show that there are infinite sequences of
weak considerate improving moves of coalitions being cliques in
$G$. In contrast, if $G$ is a disjoint set of cliques and CE reduces
to partition equilibrium, a potential function for weak (considerate)
improving moves in games with identical resources exists.

\section{Existence}
\label{sec:exist}

This section contains our main theorem showing the existence of CE in
RSGs with strictly increasing delay functions. The existence proof is
constructive and yields a polynomial time algorithm computing a state
that is both a CE and a standard NE for the RSG showing that the two
equilibrium concepts intersect.

\begin{theorem}
  \label{thm:exist}
  For any RSG with strictly increasing delay functions and any
  associated social network $G$, there exists at least one state that
  is an NE and a CE. There is a polynomial time algorithm computing
  such a state.
\end{theorem}

\begin{proof}
  We describe a process that starts in a Nash equilibrium and
  converges to a CE.  This process consists of movements of single
  players. Every strategy profile in this sequence is a standard Nash
  equilibrium.

  Consider a state $s$. Let $d_{\max}$ denote the maximal delay of a
  player in $s$. Note that in a Nash equilibrium, each used resource
  $r$ has either delay $d_r(\ell_r) = d_{\max} = d_{\max}$ or
  $d_r(\ell_r) < d_{\max}$ and $d_r(\ell_r + 1) \ge d_{\max}$. In the
  former case, we call that resource a {\em high} resource, in the
  latter case, we call it a {\em low} resource if additionally
  $d_r(\ell_r + 1) = d_{\max}$. Let $\mathcal{N}_{i,r}(s)$ denote the
  set of neighbors of player $i$ in $G$ that are on resource $r$ in
  $s$.
  We are now ready to describe the process:
  \begin{enumerate}
  \item Compute a Nash equilibrium $s$.
  \item If there is a player $i$ placed on a high resource $r$ and
    there is a low resource $r'$ with $|\mathcal{N}_{i,r}(s)| >
    |\mathcal{N}_{i,r'}(s)|$ then set $s=(s_{-i},r')$, and repeat this
    step.
  \item If there is a player $i$ placed on a high resource $r$ and
    there is a low resource $r'$ with $|\mathcal{N}_{i,r}(s)| =
    |\mathcal{N}_{i,r'}(s)|$ and $d_{r}(\ell_{r}(s) - 1) <
    d_{r'}(\ell_{r'}(s))$ then set $s=(s_{-i},r')$, and continue with
    step 2.
  \item Output $s$.
  \end{enumerate}
  Note that each state produced by this process is a Nash
  equilibrium. During this process, the following potential
  function
  $$\phi(s) = \sum_{i \in N} M |\mathcal{N}_{i}(s)| + \sum_{r
    \in R} d_r(\ell_r(s))$$ 
  decreases strictly from step to step, where we use
  $\mathcal{N}_{i}(s) = \mathcal{N}_{i,s_i}(s)$ as a shorthand for the
  neighbors of $i$ on the same resource and assume $M > \sum_{r \in R}
  d_r(n)$. One can easily modify the delay functions such that $M = n
  |R|^2$ without changing the players' preferences which implies that
  the process terminates after polynomially many steps.

  To prove that this process results in a CE, we show that if a state
  $s$ is a NE and there exists weak considerate improving move $s'_C$
  then there is also a move of a single player $i \in C$ as described
  above.

  Let $H$ and $L$ denote the set of high and low resources in $s$,
  respectively.  Let $R_h$ be the set of resources that are high in
  $s$ but no longer high in $(s'_C,s_{-C})$, and let $R_l$ be the set
  of resources that are low in $s$ and become high in $(s'_C,s_{-C})$.
  By definition, $R_h \subseteq H$ and $R_l \subseteq L$.  Let $N_h$
  be the set of players of $C$ on resources of $R_h$ in $s$, and let
  $N_l$ be the set of players of $C$ on resources of $R_l$ in $s$.

  \begin{lemma}\label{lemma:move}
    During the move $s'_C$, all players in $N_l$ are moving from
    resources in $R_l$ to resources outside of $R_l$. In turn, $|N_l|
    + |R_l|$ players move from resources in $H$ to the resources in
    $R_l$.  At least $|N_l| + |R_l|$ players are leaving $R_h$ towards
    resources outside of $R_h$.
  \end{lemma}

  \begin{proof}
    Since $s'_C$ is a weak considerate improving move, all players in
    $N_l$ are moving from resources in $R_l$ to resources outside of
    $R_l$ as their delay would increase, otherwise. These players can
    only be replaced by players of $H$ as other players would have an
    increased delay after the move, otherwise.  In turn, altogether
    $|N_l| + |R_l|$ players need to move from $H$ to $R_l$ so that the
    resources of $R_l$ become high resources after the move.
    Furthermore, we observe that the number of players on resources in
    $H \setminus R_h$ does not change during the considered move, and
    there are no players entering $H \setminus R_h$ from outside of
    $H$ as such players would have an increased delay, otherwise.  As
    a consequence, there must be at least $|N_l| + |R_l|$ players that
    are leaving $R_h$ towards $H \setminus R_h$ or $R_l$ in order to
    have $|N_l| + |R_l|$ players that move from $H$ to $R_l$.  This
    proves Lemma~\ref{lemma:move}.\qed
  \end{proof}

  The lemma implies
  \begin{eqnarray}\label{eqn4}
    |N_h| \ge |N_l| + |R_l| \enspace .
  \end{eqnarray}
  Let $\max_h = \max_{i\in N_h} {\cal N}_{i}(s)$ denote the maximum
  number of neighbors that a player of $N_h$ has on his resource. The
  definition $\max_h$ implies
  \begin{eqnarray}\label{eqn1}
    |N_h| \le (\text{max}_h+1) \cdot |R_h| \enspace .
  \end{eqnarray}
  Note that no player of $C$ has a neighbor that has chosen a resource
  from $R_l$ and is not in $C$.  Otherwise, this neighbor's delay
  would increase during the move so that $s'_C$ would not be a
  considerate move.  Therefore, we can set $\min_l = \min_{i\in N_h, r
    \in R_l} {\cal N}_{i,r}(s)$, where the choice of $i$ is
  irrelevant. The definition of $\min_l$ immediately implies
  \begin{eqnarray}\label{eqn2}
    |N_l| \ge \text{min}_l \cdot |R_l| \enspace .
  \end{eqnarray}
  Let us derive some more helpful equations regarding the different
  kinds of resources.  For each resource that decreases its load
  during the improving move, there is at least one resource that
  increases its load by one because the number of players on each low
  resource can only increase by one.  This gives
  \begin{eqnarray}\label{eqn3}
    |R_h| \le |R_l| \enspace .
  \end{eqnarray}
  Combining the Equations \ref{eqn1}, \ref{eqn4}, and \ref{eqn2} gives
  \begin{eqnarray}\label{eqn5}
    (\text{max}_h+1) \cdot |R_h| \ge |N_h| \ge |N_l| + |R_l| \ge (\text{min}_l+1) \cdot |R_l|  \enspace .
  \end{eqnarray}
  Now, we distinguish between the following two cases.

  \paragraph {Case 1:} $\max_h > \min_l$. In this case, we can set $i
  = \arg\max_{j\in N_h} {\cal N}_{j}(s)$ and $r' = \arg\min_{r \in
    R_l} {\cal N}_{i,r}(s)$, which satisfies the conditions of step 2
  of the process.

  \paragraph{Case 2:} $\max_h \le \min_l$.  In this case,
  Equation~\ref{eqn5} yields $|R_h| \ge |R_l|$, which, coupled with
  Equation~\ref{eqn3}, implies $|R_h| = |R_l|$. Substituting this
  equality back into the Equation~\ref{eqn5} gives $\max_h \ge \min_l$
  which implies $\max_h = \min_l$.  Define $q = |R_h| = |R_l|$ and $k
  = \max_h = \min_l$.  Now Equations~\ref{eqn1} and~\ref{eqn2} yield
  $|N_h| \le |N_l| + q$, which in combination with Equation~\ref{eqn4}
  yields $|N_h| = |N_l| + q$.

  On average, the resources in $R_l$ hold $|N_l|/q$ players from $C$
  in state $s$ and the resources $R_h$ hold $|N_h|/q$ players from
  $C$.  We claim that this implies that each resource in $R_l$ holds
  exactly $|N_l|/q$ players from $C$; and each resource in $R_h$ holds
  exactly $|N_h|/q$ players from $C$ and no additonal neighbour of one
  of them.  To see this, let $r_h$ denote a resource from $R_h$
  holding a maximum number of players from $C$ and let $r_l$ denote a
  resource from $R_l$ holding a minimum number of players from
  $C$. Let $i \in N_h$ be a player assigned to $r_h$. As $s'_C$ is a
  considerate move, $i$ does not have neighbors outside of $C$ on
  $r_l$. Thus, if the claim above would not hold, $i$ would have
  either at least $|N_h|/q$ neighbors on $r_h$ or strictly less than
  $|N_h|/q-1 = |N_l|/q$ neighbors on $r_l$, which would imply $\max_h
  > \min_l$ and thus contradict our assumption. As a consequence,
  $|\mathcal{N}_{i,r}(s)| = k = |\mathcal{N}_{i,r'}(s)|$, for every $i \in
  N_h$, $r \in R_h$, and $r' \in R_l$.

  Now Lemma~\ref{lemma:move} yields that each of the $q$ resources in
  $R_l$ is left by its $k$ players from $C$ and each of the $q$
  resources in $R_h$ is left by its $k+1$ players from $C$.

  We make a few further observations: The definition of $R_h$ implies
  that the number of players on a resource from $H \setminus R_h$ does
  not decrease during the considered move. Besides, this number cannot
  increase due to a weak improving move. Next consider a resource $r
  \not\in H \cup R_l$. The definition of $R_l$ implies that the number
  of players on $r$ cannot increase during a weak improving move. Now
  suppose the number of players on $r$ would decrease. Then there is a
  leaving player $i$, who moves to either $R_h$ or another resource in
  $L \setminus R_l$, as its delay would increase, otherwise. In the
  latter case, a different player must make room for $i$. By following
  this player, we can iteratively construct a chain of moving players
  until finally there is a player that moves to a resource in $R_h$.
  Thus, together with the players leaving the resources in $R_l$ there
  are at least $q k +1$ that need to migrate to a resource with a
  delay of less than $d_{\max}$ (after the move). However, the
  resources in $R_h$ have only a capacity for taking $qk$ many of such
  players. Hence, the number of players on any resource outside of
  $R_l$ or $R_h$ does not change during the considered move.

  Now consider one of the players from $N_l$. During the considered
  move, this player migrates to another resource having a delay
  strictly less than $d_{\max}$ (after the considered move). If this
  resource does not belong to $R_h$ then another player needs to leave
  this resource in order to compensate for the arriving player. Now we
  follow that player and, iteratively, construct a chain of moving
  players leading from a resource in $R_l$ to a resource in $R_h$. In
  this manner, we can decompose the set of moving players into a
  collection of $qk$ many chains each of which leads from $R_l$ to
  $R_h$. As we are considering a weak improving move the delays in
  each of these chains does not increase and there is at least one
  such chain leading from a resource $r' \in R_l$ to a resource $r \in
  R_h$ with $d_{r}(\ell_{r}(s) - 1) < d_{r'}(\ell_{r'}(s))$.  We
  choose an arbitrary player $i \in N_h$ assigned to resource $r$ in
  $s$. We have shown above that, for this player, it holds
  $|\mathcal{N}_{i,r}(s)| = |\mathcal{N}_{i,r'}(s)|$. Thus, player $i$
  satisfies the condition in step 3 of our process, which completes
  our analysis for Case 2.  \medskip

  This shows that, when the process terminates, there is no weak
  considerate improving move.  Therefore, the resulting state is an
  CE. \qed
\end{proof}

\section{Convergence}
\label{sec:converge}

Next we show that the dynamics of weak considerate improving moves by
general cliques does not have the finite improvement property, i.e.,
the dynamics corresponding to CE might cycle
(Theorem~\ref{thm:cycle}). Our construction works even for resources
with identical delays. This separates considerate equilibrium from
partition equilibrium as, in the same setting, the dynamics
corresponding to partition equilibrium admits the finite improvement
property (Proposition~\ref{thm:weaklyAcyclic}).
\begin{theorem}
  \label{thm:cycle}
  There are symmetric RSGs with strictly increasing and identical
  delays and starting states, for which there are infinite sequences
  of weak considerate improving moves by coalitions that are cliques
  in $G$.
\end{theorem}
\begin{proof}
  For the proof we construct a game with a modular structure. Our game
  consists of a number of smaller games, referred to as
  \emph{blocks}. Each block consists of 14 players and 5 resources,
  and by itself it is acyclic. However, by creating social ties across
  blocks, we create larger cliques that are able to perform
  ``resets'' in one block while making improvements in other
  blocks. By a careful scheduling of such reset moves we construct an
  infinite sequence of moves.
  
  More formally, we have 19 blocks, and in each block $i$, we have 14
  players. There are 8 players $B^i,C^i,D^i,E^i,F^i,G^i,P^i,Q^i$
  involved in our sequence, while 6 additional ``dummy'' players never
  move. The dummy players are singleton nodes in the social network
  and are only required to, in essence, simulate non-identical
  resources by increasing some of the delays to larger values. The
  social graph consists of internal links within each block and
  inter-block connections as follows. For each block, there are edges
  $\{B^i,F^i\}$, $\{C^i,E^i\}$ and $\{D^i,G^i\}$. In addition, for
  each $i=1,...,19$ there are two inter-block cliques,
  \begin{itemize}
  \setlength{\itemsep}{0pt}
  \setlength{\parsep}{0pt}
  \item $\{D^i,P^i,P^{i+1},B^{i+1},D^{i+2},P^{i+2},C^{i+6},E^{i+6}\}$ and
  \item $\{D^i,Q^i,Q^{i+1},C^{i+1},D^{i+2},Q^{i+2},B^{i+9},F^{i+9}\}$,
  \end{itemize}
  where the exponent is meant to cycle through the numbers 1 to 19,
  i.e., above $P^j$ means $P^{((j-1) \text{ mod } 19)+1}$.

  The 95 resources are denoted by $r^i_j$ with $i=1,\ldots,19$,
  $j=1,\ldots,5$. The delay functions are identical $d_r(x) = x$ for
  all $r \in R$. Note that in general, our example does not require
  linear delays, it suffices to ensure $d_r(3) > d_r(2)$.

  Let us consider a single block $i$ and a sequence of six states
  within this block depicted in Fig.~\ref{fig:blockSequence}.
  \begin{figure}[h]
    \begin{center}
      \begin{tabular}{ c | c c c | c c}
        & \hspace{0.3cm} $r_1^i$ \hspace{0.3cm} & \hspace{0.3cm} $r_2^i$ \hspace{0.3cm} & \hspace{0.3cm} $r_3^i$ \hspace{0.3cm} & \hspace{0.3cm} $r_4^i$ & \hspace{0.3cm} $r_5^i$ \hspace{0.3cm} \\
        \hline \hline
        & $C^i$ & $B^i$ &       & $P^i$ & $Q^i$ \\
        \hspace{0.3cm} $\alpha$ \hspace{0.3cm}   & $E^i$ & $D^i$ & $F^i$ & x     & x     \\
        & x     & $G^i$ & x     & x     & x     \\
        \hline
        & $C^i$ &       & $D^i$   & $P^i$ & $Q^i$ \\
        $\beta$    & $E^i$ & $B^i$ & $F^i$ & x     & x     \\
        & x     & $G^i$ & x     & x     & x     \\
        \hline
        &       & $C^i$   & $D^i$   & $P^i$ & $Q^i$ \\
        $\gamma$   & $E^i$ & $B^i$ & $F^i$ & x     & x     \\
        & x     & $G^i$ & x     & x     & x     \\
        \hline \hline
        &       & $C^i$ & $B^i$   & $P^i$ & $Q^i$ \\
        $\delta$   & $E^i$ & $D^i$ & $F^i$ & x     & x     \\
        & x     & $G^i$ & x     & x     & x     \\
        \hline
        & $D^i$ &       & $B^i$   & $P^i$ & $Q^i$ \\
        $\epsilon$ & $E^i$ & $C^i$ & $F^i$ & x     & x     \\
        & x     & $G^i$ & x     & x     & x     \\
        \hline
        & $D^i$ & $B^i$ &       & $P^i$ & $Q^i$ \\
        $\zeta$    & $E^i$ & $C^i$ & $F^i$ & x     & x     \\
        & x     & $G^i$ & x     & x     & x     \\
        \hline \hline
        & $C^i$ & $B^i$ &       & $P^i$ & $Q^i$ \\
        $\alpha$   & $E^i$ & $D^i$ & $F^i$ & x     & x     \\
        & x     & $G^i$ & x     & x     & x     \\         
      \end{tabular}
      \caption{\label{fig:blockSequence} Sequence of six states within
        a block $i$ that are attained during an infinite sequence of
        weak considerate improving moves.}
    \end{center}
  \end{figure}
  Note that $\alpha\rightarrow\beta$ represents a weak considerate
  improving move for $\{D^i,G^i\}$, where $D^i$ performs the move, and
  $G^i$ strictly improves. Similarly, $\beta\rightarrow\gamma$ is a
  weak considerate improving move for $\{C^i,E^i\}$,
  $\delta\rightarrow\epsilon$ for $\{D^i,G^i\}$, and
  $\epsilon\rightarrow\zeta$ for $\{B^i,F^i\}$. The steps $\gamma \to
  \delta$ and $\zeta \to \alpha$ are resets, in which a cyclic switch
  is performed and no player within the block strictly improves. It
  suffices to show that these steps can be implemented with moves by
  inter-block cliques.

  Consider the first reset $\gamma \to \delta$, in which $D^i$ and
  $B^i$ swap places, and for simplicity assume w.l.o.g.\ that
  $i=5$. This swap is executed in three moves, where we first swap in
  $P^5$ for $D^5$, then swap $P^5$ and $B^5$ and finally swap out
  $P^5$ to bring $D^5$ back in. This cyclic switch is the result of
  the following sequence of weak considerate improving moves: (1)
  coalition $\{D^3,P^3,P^4,B^4,D^5,P^5,C^9,E^9\}$ applies a deviation
  where $D^5$ and $P^5$ exchange their places, and $C^9$ moves away
  from $E^9$ in block 9 as $\beta\rightarrow\gamma$ prescribes; (2)
  coalition $\{D^4,P^4,P^5,B^5,D^6,P^6,C^{10},E^{10}\}$ improves by
  swapping $P^5$ and $B^5$, and moving $C^{10}$ away from $E^{10}$ in
  block 10; (3) finally, $D^5$ and $P^5$ swap with coalition
  $\{D^5,P^5,P^6,B^6,D^7,P^7,C^{11},E^{11}\}$ where $C^{11}$ moves
  away from $E^{11}$ in block 11. In the final dynamics, we will use
  these moves also to simultaneously perform swaps in the other blocks
  3, 4, 6, and 7.

  The second reset swap $\zeta \to \alpha$ by $D^5$ and $C^5$ can be
  done in similar fashion by a circular swap involving $Q^5$ and using
  the $B^i$ and $F^i$ players of blocks $i=12,13,14$. Note that our
  edges are carefully designed not to generate any undesired
  connections. In particular, $D^5$, $P^5$, $B^5$ rely on the movement
  of $C^9$, $C^{10}$ and $C^{11}$ to execute their swaps. During these
  swaps, $B^9$, $B^{10}$ and $B^{11}$ are deteriorated. None of the
  deteriorated players are attached to players in the respective
  improving coalitions, i.e., none of $D^3$, $P^3$, $P^4$, $B^4$,
  $D^5$ or $P^5$ are friends with $B^9$, none of $D^4$, $P^4$, $P^5$,
  $B^5$, $D^6$ or $P^6$ are friends with $B^{10}$, and none of $D^5$,
  $P^5$, $P^6$, $B^6$, $D^7$ or $P^7$ are friends with $B^{11}$. In
  addition, for making the switch between $D^5$, $Q^5$ and $C^5$ we
  use the movement of $B^{12}$, $B^{13}$ and $B^{14}$. Note that none
  of the players required to execute the switches are friends with
  $C^{12}$, $C^{13}$ or $C^{14}$, respectively.

  An infinite sequence of weak considerate improving moves can now,
  for example, be obtained from a starting state as follows. We indicate
  for each block in which state $\alpha$ to $\zeta$ it is
  initialized. Here $\gamma_1$, $\gamma_2$, $\zeta_1$, and $\zeta_2$
  indicate the intermediate states of the corresponding circular
  resetting swaps.
  \[
  \begin{tabular}{c|c|c|c|c|c|c|c|c|c|c|c|c|c|c|c|c|c|c}
    \hspace{0.02cm}
    1 \hspace{0.02cm} & \hspace{0.02cm} 2 \hspace{0.02cm} & \hspace{0.02cm} 3 \hspace{0.02cm} & \hspace{0.02cm} 4 \hspace{0.02cm} & \hspace{0.02cm} 5 \hspace{0.02cm} & \hspace{0.02cm} 6 \hspace{0.02cm} & \hspace{0.02cm} 7 \hspace{0.02cm} & \hspace{0.02cm} 8 \hspace{0.02cm} & \hspace{0.02cm} 9 \hspace{0.02cm} & \hspace{0.02cm} 10 \hspace{0.02cm} & \hspace{0.02cm} 11 \hspace{0.02cm} & \hspace{0.02cm} 12 \hspace{0.02cm} & \hspace{0.02cm} 13 \hspace{0.02cm} & \hspace{0.02cm} 14 \hspace{0.02cm} & \hspace{0.02cm} 15 \hspace{0.02cm} & \hspace{0.02cm} 16 \hspace{0.02cm} & \hspace{0.02cm} 17 \hspace{0.02cm} & \hspace{0.02cm} 18 \hspace{0.02cm} & \hspace{0.02cm} 19 \hspace{0.02cm}\\
    \hline 
    $\zeta_2$ & $\zeta_1$ & $\zeta$ & $\zeta$ & $\zeta$ & $\zeta$ & $\zeta$ & $\zeta$ & $\zeta$ & $\epsilon$ & $\delta$ & $\gamma_2$ & $\gamma_1$ & $\gamma$ & $\gamma$ & $\gamma$ & $\gamma$ & $\beta$ & $\alpha$
  \end{tabular}
  \]
  In the first step, we can simultaneously advance blocks 1-3 from
  $(\zeta_2, \zeta_1, \zeta)$ to $(\alpha, \zeta_2, \zeta_1)$ using
  movement of $B^{10}$, which advances block 10 to $\zeta$. In the
  next step we advance blocks 12-14 from $(\gamma_2,\gamma_1,\gamma)$
  to $(\delta, \gamma_2,\gamma_1)$ using movement of $C^{18}$, which
  advances block 18 to $\gamma$. Next, we make two internal switches
  in blocks 11 from $\delta$ to $\epsilon$ and 19 from $\alpha$ to
  $\beta$. In this way, we have shifted the state sequence by one
  block, which implies that we can repeat this sequence
  endlessly. \qed
\end{proof}

In contrast, observe that if the graph is a set of disjoint cliques,
then for games with identical and strictly increasing delay function
we can easily construct a potential function showing existence and
acyclicity with respect to weak (considerate) improving moves.
\begin{proposition}
  \label{thm:weaklyAcyclic}
  In every symmetric RSG with strictly increasing, identical delays
  every sequence of weak improving moves of allowed partitions is
  finite and ends in a partition equilibrium.
\end{proposition}
Note that in this case we can assume w.l.o.g.\ that $d_r(x) = x$ for
all $r \in R$. Also, each weak improving move decreases the sum of
costs of all players in the partition. Thus, the results
of~\cite{Fotakis08} for linear delays directly imply the finite
improvement property.

\bibliographystyle{plain}
\bibliography{../../Bibfiles/literature}

\begin{thebibliography}{10}

\bibitem{AnshelevichSAGT10}
Elliot Anshelevich, Bugra Caskurlu, and Ameya Hate.
\newblock Partition equilibrium always exists in resource selection games.
\newblock In {\em Proc.\ 3rd Intl.\ Symp.\ Algorithmic Game Theory (SAGT)},
  pages 42--53, 2010.

\bibitem{Ashlagi08}
Itai Ashlagi, Piotr Krysta, and Moshe Tennenholtz.
\newblock Social context games.
\newblock In {\em Proc.\ 4th Intl.\ Workshop Internet \& Network Economics
  (WINE)}, pages 675--683, 2008.

\bibitem{Aumann59}
Robert Aumann.
\newblock Acceptable points in general cooperative n-person games.
\newblock In {\em Contributions to the Theory of Games IV}, volume~40 of {\em
  Annals of Mathematics Study}, pages 287--324. Princeton University Press,
  1959.

\bibitem{Brandt06}
Felix Brandt, Felix Fischer, and Yoav Shoham.
\newblock On strictly competitive multi-player games.
\newblock In {\em Proc.\ 21st Conf.\ Artificial Intelligence (AAAI)}, pages
  605--612, 2006.

\bibitem{Feldman09}
Michal Feldman and Moshe Tennenholtz.
\newblock Partition equilibrium.
\newblock In {\em Proc.\ 2nd Intl.\ Symp.\ Algorithmic Game Theory (SAGT)},
  pages 48--59, 2009.

\bibitem{Fotakis09}
Dimitris Fotakis, Spyros Kontogiannis, Elias Koutsoupias, Marios Mavronicolas,
  and Paul Spirakis.
\newblock The structure and complexity of {N}ash equilibria for a selfish
  routing game.
\newblock {\em Theoret.\ Comput.\ Sci.}, 410(36):3305--3326, 2009.

\bibitem{Fotakis08}
Dimitris Fotakis, Spyros Kontogiannis, and Paul Spirakis.
\newblock Atomic congestion games among coalitions.
\newblock {\em ACM Trans.\ Algorithms}, 4(4), 2008.

\bibitem{Hayrapetyan06}
Ara Hayrapetyan, {\'E}va Tardos, and Tom Wexler.
\newblock The effect of collusion in congestion games.
\newblock In {\em Proc.\ 38th Symp.\ Theory of Computing (STOC)}, pages 89--98,
  2006.

\bibitem{Holzman97}
Ron Holzman and Nissan Law-Yone.
\newblock Strong equilibrium in congestion games.
\newblock {\em Games Econom.\ Behav.}, 21(1-2):85--101, 1997.

\bibitem{Holzman03}
Ron Holzman and Nissan Law-Yone.
\newblock Network structure and strong equilibrium in route selection games.
\newblock {\em Math.\ Social Sci.}, 46(2):193--205, 2003.

\bibitem{Kuniavsky07}
Sergey Kuniavsky and Rann Smorodinsky.
\newblock Coalitional congestion games.
\newblock Master's thesis, Technion, Haifa, Israel, 2007.

\bibitem{Monderer96}
Dov Monderer and Lloyd Shapley.
\newblock Potential games.
\newblock {\em Games Econom.\ Behav.}, 14:1124--1143, 1996.

\bibitem{Myerson04}
Robert Myerson.
\newblock {\em Game Theory: Analysis of Conflict}.
\newblock Harvard University Press, 6th edition, 2004.

\bibitem{Nash51}
John Nash.
\newblock Non-cooperative games.
\newblock {\em Annals of Mathematics}, 54(2):286--295, 1951.

\bibitem{Rosenthal73}
Robert Rosenthal.
\newblock A class of games possessing pure-strategy {N}ash equilibria.
\newblock {\em Intl.\ J. Game Theory}, 2:65--67, 1973.

\bibitem{Rozenfeld07}
Ola Rozenfeld.
\newblock Strong equilibrium in congestion games.
\newblock Master's thesis, Technion, Haifa, Israel, 2007.

\end{thebibliography}

\end{document}